\documentclass[11pt,a4paper,english]{amsart}
\usepackage{babel}
\usepackage[latin1]{inputenc}
\usepackage{amsmath}
\usepackage{amsthm}
\usepackage{amsfonts}
\usepackage{indentfirst}
\usepackage{graphicx}
\usepackage{caption}
\usepackage{amssymb}
\usepackage[mathscr]{eucal}
\usepackage{tikz}

\usepackage{amsthm}
\usepackage{amscd}
\newtheorem{pro}{Proposition}
\newtheorem{cor}{Corollary}
\newtheorem{teo}{Theorem}
\newtheorem{rem}{Remark}
\newtheorem{lem}{Lemma}

\def\a{\alpha}

\title[New quantum codes from evaluation and  matrix-product codes]{New quantum codes from evaluation and  matrix-product codes}
\author[Galindo, Hernando, Ruano]{{\small Carlos Galindo, Fernando Hernando}\\{\tiny Instituto
Universitario de Matem\'aticas y Aplicaciones de Castell\'on\\and
Departamento de Matem\'aticas\\ Universitat Jaume I, Campus de Riu
Sec. 12071 Castell\'{o} (Spain)} \vspace{1mm}\\
and {\small Diego Ruano}\\{\tiny Department of Mathematical Sciences, Aalborg University\\ Fredrik Bajers Vej 7G, 9220 Aalborg East (Denmark)}}
\email{{\rm Galindo:} galindo@uji.es; {\rm Hernando:} carrillf@uji.es; {\rm Ruano:} diego@math.aau.dk}
\date{}
\thanks{Supported by the Spanish Ministry of Economy:
grant MTM2012-36917-C03-03, the University Jaume I: grant PB1-1B2012-04 and the Danish Council for Independent Research, grant DFF-4002-00367.}

\begin{document}

\begin{abstract}
Stabilizer codes  obtained via the CSS code construction and the Steane's enlargement of subfield-subcodes and matrix-product codes coming from generalized Reed-Muller, hyperbolic and affine variety codes are studied. Stabilizer codes with good quantum parameters are supplied, in particular, some binary codes of lengths $127$ and $128$ improve the parameters of the codes in  {\tt http://www.codetables.de}. Moreover, non-binary codes are presented either with parameters better than or equal to the quantum codes obtained from BCH codes by La Guardia or with lengths that cannot be reached by them.
\end{abstract}


\maketitle



\section{Introduction}
Quantum computers are based on the principles of quantum mechanics and use subatomic particles (qubits) to hold memory. The construction of efficient devices  of this type  would have important consequences as the breaking of some well-known cryptographical schemes \cite{22RBC}. Information on a recent attempt to built a quantum computer can be found in \cite{BCM, SS}.

Despite quantum mechanical systems are very sensitive to disturbances and arbitrary quantum states cannot be replicated, error correction is possible \cite{23RBC}. In this paper we are concerned with stabilizer codes which are a class of quantum error-correcting codes. Parameters of our codes will be expressed as $[[n,k,d]]_q$, where $q$ is a power $p^r$ of  a prime number $p$  and $r$ a positive integer. They mean that our codes are $q^k$-dimensional linear subspaces of  $\mathbb{C}^{q^n}$, $\mathbb{C}$ being the complex field, and $d$ its minimum distance, which determines detection and correction of errors. An stabilizer code is called to be pure to a positive integer $t$ whenever its stabilizer group does not contain non-scalar matrices with weight less than $t$ (see, for instance, \cite{18RBC, kkk} for details).

Stabilizer codes can be derived from classical ones with respect to Symplectic or Hermitian inner product \cite{19kkk,BE,AK,kkk}, although this can also be done with respect to Euclidean inner product by using the so-called CSS code construction \cite{20kkk,95kkk}. The following results \cite[Lemma 20 and Corollary 21]{kkk} show the parameters of the stabilizer codes that one gets by using the above mentioned code construction. The reader can consult \cite[Theorem 13]{kkk} to see how stabilizer codes are obtained from classical ones.

\begin{teo}
\label{antescoro}
Let $C_1$ and $C_2$ two linear error-correcting block codes with parameters $[n,k_1,d_1]$ and  $[n,k_2,d_2]$ over the field  $\mathbb{F}_{q}$ and such that $C_2^{\perp} \subset C_1$, where $C_2^{\perp}$ stands for the dual code of $C_2$.  Then, there exists an $[[n, k_1+k_2 -n,  d]]_q$ stabilizer code with minimum distance $$d = \min \left\{\mathrm{w}(c) | c \in (C_1 \setminus C_2^\perp) \cup (C_2 \setminus C_1^\perp) \right\},$$ which is pure to $\min \{d_1,d_2\}$, where $\mathrm{w}(c)$ denotes the weight of a word $c$.
\end{teo}

\begin{cor} \label{bueno}
Let $C$ be a linear $[n,k,d]$ error-correcting block code over  $\mathbb{F}_{q}$ such that  $C^\perp \subset  C$. Then, there exists an $[[n, 2k -n, \geq d]]_q$ stabilizer code which is pure to $d$.
\end{cor}

Note that we use the symbol $\subset$ to indicate subset, in particular $C \subset C$ holds.  Next, we state the  Hamada's generalization of the Steane's enlargement procedure \cite{Steane2} because it will be used in the paper. Given two suitable codes $C$ and $C'$, the code obtained by applying this procedure will be called their Steane's enlargement and denoted by $\mathrm{SE}(C,C')$.

\begin{cor} \label{ham} \cite{ham}
Let $C$ be an $[n,k]$ linear code over the field $\mathbb{F}_q$ such that $C^\perp \subset C$. Assume that $C$ can be enlarged to an $[n,k']$ linear code $C'$, where $k' \geq k + 2$. Then, there exists a stabilizer code with parameters $[[n,k+k'-n, d \geq \min \{d', \lceil \frac{q+1}{q} d'' \rceil\} ]]_q$, where $d' = \mathrm{w} (C \setminus C'^\perp)$, $d'' = \mathrm{w} (C' \setminus C'^\perp)$ and  $\mathrm{w}$ denotes the minimum weight of the words of a set.
\end{cor}

In this paper, we provide new families of algebraically generated stabilizer codes derived essentially from the Euclidean inner product and containing a number of codes with good parameters. In fact we are able to  improve some of the binary quantum  codes given in \cite{codet}. Moreover we supply stabilizer codes with parameters better than or equal to those given in \cite[Table III]{lag3} and \cite{lag1}, together with others whose lengths cannot be reached in \cite{lag3,lag1} but  exceed the Gilbert-Varshamov bounds \cite{eck, mat, feng}, \cite[Lemma 31]{kkk}, or improve those in \cite{edel} or satisfy both conditions.

Our stabilizer codes are supported in three families of linear codes: the so-called  (generalized) Reed-Muller codes \cite{52hlp, 15hlp}, hyperbolic (or hyperbolic cascaded Reed-Solomon) codes \cite{FR,kaba,massey,saint}, and affine variety codes \cite{FL, Geil-Affine}.  Stabilizer codes obtained with Reed-Muller codes have been studied in \cite{Steane} and \cite{Sarvepalli}. These families allow us to get nested sequences of codes that contain their dual ones and determine stabilizer codes by applying the CSS code construction. We enlarge this set of stabilizer codes by considering suitable matrix-product codes, which were introduced in \cite{bn} (see also \cite{8hlr, hlr}). Together with this construction, we also  consider  subfield-subcodes of our codes, which allows us to get codes over small fields from codes defined over larger ones, always within the same characteristic. We complement the mentioned techniques with  the Steane's enlargement (Corollary \ref{ham}).

Tables with parameters of our codes are distributed along the paper and testify their goodness. As mentioned, several of them improve the parameters available in the literature. These tables are presented as a complement of the different procedures described for obtaining our families of stabilizer codes. Reed-Muller and hyperbolic codes have the advantage that all their parameters are known and, as a consequence, we are able to obtain, with a simple calculation, parameters for the corresponding stabilizer codes. Affine variety codes give a broader spectrum of codes and their dimensions and lengths can be computed. Unfortunately, there is no known general formula for their distances and it seems a very difficult problem to obtain it. In this paper, we have computed them by using the computational algebra system Magma \cite{magma}.

We finish this introduction with a description of each section in the paper. Section \ref{matrix} reviews the concept of matrix-product code and recalls two useful results, Theorem \ref{fer1} and Corollary \ref{coro2}. It is worthwhile to mention that using non-singular by columns matrices is a key point
for obtaining matrix-product codes with good parameters that provide stabilizer codes by Corollary \ref{bueno}. Our supporting families of codes are introduced in Section \ref{families}. Parameters for them and conditions for self-orthogonality are the main facts we state there. In particular, Theorem \ref{te: ConditionSelfOrthogonal} has interest in classical coding theory as well, since it characterizes self-orthogonal hyperbolic codes. Subfield-subcodes of the mentioned families are also an important tool for constructing our stabilizer codes and we devote Section \ref{subfield} to give our main results in this line. The main theoretical results in this paper can be found in Section \ref{quantumm}. Indeed, Theorems \ref{eseldoce} and \ref{UNO} together with Remark \ref{U} give parameters for the stabilizer codes obtained from the previous constructions. The remaining sections of the paper show tables with quantum parameters of codes obtained as we have described.  For certain small sizes, there are no non-singular by columns orthogonal matrices  over the fields $\mathbb{F}_2$ and $\mathbb{F}_3$. To avoid this difficulty, Theorems \ref{te: QMPCsobreF2} and \ref{te: QMPCsobreF3} in Sections \ref{efe2} and \ref{efe3} look for clever matrices that allow us to get self-orthogonal matrix-product codes. Codes over $\mathbb{F}_3$, improving some ones in \cite{lag1}, are treated in Section \ref{efe3}  while Section \ref{efe2} is devoted to binary ones, where, together with codes derived from Theorem  \ref{te: QMPCsobreF2}, in Table \ref{ta:best} we show stabilizer codes of lengths $127$ and $128$ improving \cite{codet}. The last section in the paper, Section \ref{no2ni3}, contains a number of  stabilizer codes over the fields $\mathbb{F}_4$, $\mathbb{F}_5$ and $\mathbb{F}_7$. In Table \ref{dd} and comparing with \cite[Table III]{lag3}, the reader will find a code improving that table, another one with a new distance and some more with the same parameters; in fact our codes can reproduce most parameters in the mentioned Table III. In addition Table \ref{dd} shows stabilizer codes either improving \cite{edel} or exceeding the Gilbert-Varshamov bounds and with lengths that cannot be reached in \cite{lag3,lag1}.

\section{Matrix-product codes}  \label{matrix}
Along this paper, $p$ is a prime number, $q=p^r$ a positive integer power of $p$ and $\mathbb{F}_q$ the finite field with $q$ elements.
 Let $C_1,C_2, \ldots, C_s$ be a family of $s$ codes of length $m$ over $\mathbb{F}_q$ and $A = (a_{ij})$ an $s \times l$ matrix with entries in $\mathbb{F}_q$. Then, the {\it matrix-product code} \cite{bn}, given by the above data and denoted $[C_1,C_2, \ldots, C_s] \cdot A$, is defined as the code over $\mathbb{F}_q$  of length $ml$ whose generator matrix is
\begin{equation}
\label{lamatriz}
\left(
  \begin{array}{ccccc}
a_{11}G_1 & a_{12}G_1 & \cdots & a_{1l}G_1 \\a_{21}G_2 & a_{22}G_2 & \cdots  & a_{2l}G_2 \\
\vdots & \vdots & \cdots & \vdots \\
a_{s1}G_s & a_{s2}G_s & \cdots & a_{sl}G_s \\
\end{array}
\right),
\end{equation}
where $G_i$, $1 \leq i \leq s$, is a generator matrix for the code $C_i$.

Given a matrix $A$ as above, let $A_t$ be the matrix consisting of the first $t$ rows of $A$. For $1\leq j_1<  j_2 < \cdots <
j_t\leq l$, we denote by $A(j_1, j_2,\ldots,j_t)$ the $t\times t$ matrix consisting of the columns $j_1, j_2,\ldots,j_t$ of $A_t$. A {\it non-singular by columns matrix} over $\mathbb{F}_q$ is a matrix $A$ satisfying that every sub-matrix $A(j_1,j_2, \ldots,j_t)$ of $A$, $1 \leq t \leq s$, is non-singular \cite{bn}.  Some of the codes in this paper are based on matrix-product codes, whose parameters are described in the following result.

\begin{teo} \label{fer1}
\cite{hlr, 8hlr} The matrix-product code $[C_1,C_2, \ldots, C_s] \cdot A$ given by a sequence of $[m,k_i,d_i]$-linear codes $C_i$ over $\mathbb{F}_q$ and a full-rank matrix $A$ is a linear code whose length is $ml$, it has dimension $\sum_{i=1}^s k_i$ and minimum distance larger than or equal to
\[
\delta  =  \min_{1 \leq i \leq s} \{d_i \delta_i\},
\]
where $\delta_i$ is the minimum distance of the code on $\mathbb{F}_q^{l}$ generated by the first $i$ rows of the matrix $A$. Moreover, when the matrix $A$ is non-singular by columns, it holds that $\delta_i=l+1-i$. Furthermore, if we assume that the codes $C_i$ form a nested sequence $C_1 \supset C_2 \supset \cdots \supset C_s$, then the minimum distance of the code $[C_1,C_2, \ldots, C_s] \cdot A$ is exactly $\delta$.
\end{teo}

More information about this  construction can be found in \cite{hlr, hr1, hr3, hr2}. Since we are interested in stabilizer codes obtained by applying the CSS code construction, the following results concerning duality will be of interest.

\begin{teo} \label{norton}
\cite{bn}
Assume that $\{C_1,C_2, \ldots, C_s\}$ is a family of linear codes of length $m$ and $A$ a non-singular $s \times s$ matrix, then the following equality of codes happens
\[
\left([C_1,C_2, \ldots, C_s] \cdot A \right)^\perp = [C_1^\perp, C_2^\perp, \ldots, C_s^\perp]\cdot \left(A^{-1} \right)^t,
\]
where, as usual, $B^t$  denotes  the transpose of the matrix $B$.
\end{teo}


\begin{cor}
\label{coro2}
Let $A$ be an orthogonal $s \times s$ matrix (i.e., a matrix such that $\left(A^{-1} \right)^t=A$) and assume that for $i=1,2, \ldots,s$, it  holds that $C_i^{\perp}\subset C_i$ then
\[
\left([C_1,C_2, \ldots, C_s] \cdot A \right)^\perp \subset [C_1,C_2, \ldots, C_s] \cdot A.
\]
\end{cor}

\section{Some families of codes and their dual ones}
\label{families}

Next, we introduce some known families of codes which we will use for our purposes.

\subsection{Reed-Muller codes} \label{reed-muller} Consider the ring of polynomials  $\mathbb{F}_{q}[X_1, X_2, \ldots, X_m]$ in $m$ variables over the field $\mathbb{F}_{q}$ and its ideal $I= \langle X_1^{q}-X_1, X_2^{q}-X_2,  \ldots, X_m^{q}-X_m \rangle$. Set $R =  \mathbb{F}_{q}[X_1, X_2, \ldots, X_m]/I$ the corresponding $\mathbb{F}_{q}$-algebra and write $Z(I)= \mathbb{F}_{q}^m =\{P_1, \ldots,P_n\}$, the set of zeroes in  $\mathbb{F}_{q}$ of the ideal $I$. We will use the evaluation map $\mathrm{ev}:  R \rightarrow \mathbb{F}_{q}^n$ defined by $\mathrm{ev}(f)=(f(P_1), \ldots,f(P_n))$ for classes of polynomials $f\in R$. It is well-known that $\mathrm{ev}$ is, in this case, an isomorphism of $\mathbb{F}_{q}$-vector spaces. For a positive integer $r$, the (generalized) {\it Reed-Muller} code of order $r$ on $\mathbb{F}_{q}[X_1, X_2, \ldots, X_m]$ (or the $(r,m)$ Reed-Muller code) is defined as $
RM(r,m)= \left\{ \mathrm{ev}(f) \;|\; f \in R, \; \deg(f) \leq r \right\}$, where $\deg$ means total degree. Notice that we always choose a canonical representative of $f$ without powers $X_i^j$, $j \geq q$ and the length of the codes is $n=q^m$. The following result summarizes known results for Reed-Muller codes (see \cite{HLP, Sarvepalli}, for instance).
\begin{teo}
\label{reed}
With the above notations, assume $0 \leq r < (q-1)m$ and by Euclidean division, set $(q-1)m - r = a (q-1) + b$, $a, b \geq 0$ and $b < q-1$. Then
\begin{enumerate}
\item  The dimension of the code $RM (r,m)$ is   $$\sum_{j=0}^m (-1)^j \binom {m} {j} \binom {m+r -jq} {r-jq}.$$
  \item The minimum distance of the code $RM(r,m)$ is $(b+1) q^a$.
  \item The dual $RM(r,m)^\perp$ of the code $RM(r,m)$ is the Reed-Muller code $RM(m(q-1)-(r+1),m)$.
\end{enumerate}
\end{teo}

In order to consider matrix-product codes with Reed-Muller ones, suitable for the CSS code construction, we will consider a positive integer $r$ such that
$2 r + 1 \leq m(q-1)$.
Now, writing $E=RM(r,m)$ and $C=RM(m(q-1)-(r+1),m)$, the equality $E=C^\perp$ happens. From the above equality, the code inclusion $C^\perp \subset C$ holds and setting $r+1= c(q-1) +e$ by Euclidean division, one gets that the  minimum distance of the code $C$ is $d=(e+1) q^{c}$ and its dimension is
\begin{equation}
\label{DIM}
k= \sum_{j=0}^m (-1)^j \binom {m} {j} \binom {(m-j)q-(r+1)} {m(q-j-1)-(r+1)}.
\end{equation}

\subsection{Hyperbolic codes} \label{hyperbolic}
The first part of this section is based  on \cite{gh}. Consider the same $\mathbb{F}_q$-algebra $R$ defined in the previous subsection and fix a positive integer $t$, $0 \leq t \leq q^m$. Define the linear code, $\Xi(t,m)$, on $\mathbb{F}_q^n$ generated by the vectors obtained by applying the map $\mathrm{ev}$ to the set of monomials:
\begin{equation}
\label{HHyp}
\mathfrak{M} := \left\{ X_1^{\alpha_1}  X_2^{\alpha_2} \cdots X_m^{\alpha_m} \; | \; 0 \leq \alpha_i < q, 1 \leq i \leq m \;  \mbox{ and } \; \prod_{i=1}^m (\alpha_i + 1) < q^m - t \right\}.
\end{equation}

The {\it $t$-th hyperbolic code} (on  $\mathbb{F}_{q}[X_1, X_2, \ldots, X_m]$), Hyp$(t,m)$ is, by definition, the dual code of the code $\Xi(t,m)$ above given.  Therefore, the length of the codes is again $n=q^m$.  For simplicity's sake, we set $X^{\boldsymbol{\alpha}}$ instead of $X_1^{\alpha_1}  X_2^{\alpha_2}\cdots X_m^{\alpha_m} $ for an element $\boldsymbol{\alpha} = (\alpha_1, \alpha_2, \ldots, \alpha_m)$ in $\mathbb{Z}^m$, $\alpha_i \geq 0$. For each element $\boldsymbol{\alpha}$ of this type, we define
\[
D_{\boldsymbol{\alpha}}  = \left\{
 X^{\boldsymbol{\beta}}\; | \; 0 \leq \beta_i < q, 1 \leq i \leq m \mbox{ and $X^{\boldsymbol{\beta}}$ is not divisible by $X^{\boldsymbol{\alpha}}$}
\right\}
\]
and $n_{\boldsymbol{\alpha}}  = \mathrm{card} D_{\boldsymbol{\alpha}}$. Without loss of generality, we can assume that the values $t$ used in the definition of hyperbolic codes are of the form $n_{\boldsymbol{\alpha}}$ for some $\boldsymbol{\alpha}$ such that $0 \leq \alpha_i < q$ for all $i$. This is assumed in the rest of the paper. To make clear the previous assumption, note that if one picks any positive integer $0 \leq s \leq q^m$, then there exists a positive integer $t \geq s$ of the form $t= n_{\boldsymbol{\alpha}}$, for some $\boldsymbol{\alpha}$ as above, such that $\mathrm{Hyp}(s,m) = \mathrm{Hyp}(t,m)$.

\begin{teo}
\label{elcinco} \cite{gh}
Consider the hyperbolic code $\mathrm{Hyp}(t,m)$ above defined, with $t= n_{\boldsymbol{\alpha}}$, for some $\boldsymbol{\alpha}$. Then,
\begin{enumerate}
  \item $\mathrm{Hyp}(t,m)$ is generated by the set of vectors in $\mathbb{F}_q^n$ obtained by applying $\mathrm{ev}$ to the set of monomials $X^{\boldsymbol{\alpha}}$ such that $n_{\boldsymbol{\alpha}}$ is less than or equal to $t$.
\item The minimum distance of $\mathrm{Hyp}(t,m)$ is $q^m -t$.

\item The dimension of $\mathrm{Hyp}(t,m)$ is $ q^m - \mathrm{card} \; \mathfrak{M}$.
\end{enumerate}

\end{teo}

Bearing Theorem \ref{elcinco}, it is not difficult to deduce that the code $\mathrm{Hyp}(t,m)$ is generated by those vectors obtained after applying ev to the set of monomials $X^{\boldsymbol{\alpha}}$, where $\boldsymbol{\alpha}$ runs over the set
\begin{equation}
\label{DosDos}
 \left\{\boldsymbol{\alpha} = (\alpha_1, \alpha_2, \ldots, \alpha_m) \in \mathbb{Z}^m \; | \; 0 \leq \alpha_i \leq q-1, 1 \leq i \leq m, \prod_{i=1}^m (q-\alpha_i) \ge q^m - t \right\}.
\end{equation}

One may consider that, given a designed minimum distance $q^m -t$, hyperbolic codes are defined by  set (\ref{HHyp}) for maximizing its dimension. It is worth to mention that the dual of a hyperbolic code is not a hyperbolic code. We desire  to know for which values $t$ the inclusion $(\mathrm{Hyp}(t,m))^{\perp}\subset \mathrm{Hyp}(t,m)$ holds. To decide it, we need  the following lemma.

\begin{lem}\label{le: t1<t2}
With the above notation, assume that $t_1\leq t_2$ then $\mathrm{Hyp}(t_1,m)\subset \mathrm{Hyp}(t_2,m)$. Moreover if $t_1 < t_2$, it holds the following equality of  minimum weights: $$\mathrm{w}\left(\mathrm{Hyp}(t_2,m)\setminus \mathrm{Hyp}(t_1,m)\right)=\mathrm{w}\left(\mathrm{Hyp}(t_2,m)\right).$$
\end{lem}
\begin{proof}
The first assertion follows simply by taking into account the sets of monomials determined by the tuples in (\ref{DosDos}) and involved in the construction of both codes. For the second one, $\mathrm{w}(\mathrm{Hyp}(t_2,m))=q^m-t_2 < q^m-t_1=\mathrm{w}(\mathrm{Hyp}(t_1,m))$, what concludes the proof.
\end{proof}
\vspace{2mm}

The following result shows when $C^\perp \subset C$ for a hyperbolic code $C$.

\begin{teo}\label{te: ConditionSelfOrthogonal} Consider the hyperbolic code $\mathrm{Hyp}(t,m)$ above defined, with $t= n_{\boldsymbol{\alpha}}$, for some $\boldsymbol{\alpha}$. Then, the codes' inclusion $(\mathrm{Hyp}(t,m))^{\perp}\subset \mathrm{Hyp}(t,m)$ happens if, and only if, one of the following conditions hold.
\begin{enumerate}
  \item The integer $m$ is even and  $t\geq q^m-q^{\frac{m}{2}}$.
  \item Both, the integer $m$ and the cardinality of the base field $q$, are odd and $t$ satisfies   $t\geq q^m-q^{\frac{m-1}{2}}\frac{q+1}{2}$.
  \item The integer $m$ is odd, the cardinality of the base field $q$ is even and $t$ satisfies $t\geq q^m-q^{\frac{m-1}{2}}(\frac{q}{2}+1)$.
\end{enumerate}
 \end{teo}
\begin{proof}
Firstly we are going to prove that, in each one of the three cases in the statement, the codes' inclusion $(\mathrm{Hyp}(t,m))^{\perp}\subset \mathrm{Hyp}(t,m)$ does not hold when the corresponding inequality does not happen. To do it, we will provide a $m$-tuple $\mathbf{a} = (a_1,a_2, \ldots, a_m)$ attached to a monomial in the set (\ref{HHyp}) which is not in the set (\ref{DosDos}).

With respect to the first case, assume that $m$ is even and $t < q^m-q^{\frac{m}{2}}$. Set $a_i=q-1$ for $i=1, 2, \ldots,m/2$ and $a_i = 0$ otherwise. Is clear that $\prod_{i=1}^m (a_i+1)=q^{m/2}<q^{m} - t $ so $X^{\mathbf{a}}$ is in the set (\ref{HHyp}). However $\prod_{i=1}^m (q-a_i)=q^{m/2} < q^{m} - t $. Therefore $\mathbf{a}$ is not in the set (\ref{DosDos}). Now, consider the second case and suppose that $m$ and $q$ are odd and $t < q^m-q^{\frac{m-1}{2}}\frac{q+1}{2}$. Then, the $m$-tuple $\mathbf{a}$ defined as $a_i=q-1$ for $i=1, 2, \ldots,(m-1)/2$, $a_m=(q-1)/2$ and $a_i=0$ otherwise satisfies the requirements. Finally, in our third case, the facts $m$ odd and $q$ even show that if $t < q^m-q^{\frac{m-1}{2}}(\frac{q}{2}+1)$, an $m$-tuple $\mathbf{a}$ satisfying the desired condition is defined by $a_i=q-1$ for $1 \leq i \leq (m-1)/2$, $a_m=q/2$ and $a_i=0$ otherwise, which concludes this part of the proof.

It remains to prove that, in each one of the previous cases, when $t$ is larger than or equal to the bounds above indicated,  the inclusion $(\mathrm{Hyp(t,m)})^{\perp}\subset \mathrm{Hyp}(t,m)$ holds. Before carrying on with the technical details, we notice that Lemma   \ref{le: t1<t2} proves that if $t_1 \leq t_2$, then $\mathrm{Hyp}(t_1,m)\subset \mathrm{Hyp}(t_2,m)$  and, moreover, $(\mathrm{Hyp}(t_1,m))^{\perp}\supset (\mathrm{Hyp}(t_2,m))^{\perp}$. Therefore it is enough to prove the remaining part of our theorem in the mentioned cases and when $t$ coincides with our bounds.

Consider the hypercube $\mathfrak{H}$ of rational points $(x_1,x_2, \ldots, x_m)$ in $\mathbb{Z}^m$ such that $0 \leq x_i \leq q-1$ for $1 \leq i \leq m$, i.e.,  $\mathfrak{H} = (\{0,1,\ldots,q-1 \})^m$, and the varieties on $\mathbb{R}^m$, $H_1$ and $H_2$, defined, respectively, by the equations $(X_1+1) (X_2+1) \cdots(X_m+1)=q^m-t$ and $(q-X_1) (q-X_2) \cdots(q-X_m)=q^m-t$. From the above considerations, it is clear that to prove our result, we must check the following condition that we denote by (*): All rational point in $\mathfrak{H}$ under the variety $H_1$ must be under or on the variety $H_2$. Notice that the expression under (respectively, under or on) $H$ means rational points in the space bounded by the hyperplanes $X_i=0$ and $H$ and containing the zero vector, which also can belong to the hyperplanes but not to (respectively, and) the variety $H$.

The conjugation map is defined on the closure $\overline{\mathfrak{H}}$ of $\mathfrak{H}$ as $\varphi: \overline{\mathfrak{H}} \rightarrow \overline{\mathfrak{H}}$, $\varphi(x_1,x_2, \ldots, x_m) = (q-1-x_1, q-1-x_2, \ldots, q-1-x_m)$ and will help us in our reasoning. For a start, it is straightforward to check that $\varphi(H_1) = H_2$. Now, as we announced, we are going to prove Condition (*) for each case.

In case (1), it happens that on $\overline{\mathfrak{H}}$, the hyperplane $\pi$ with equation $X_1+X_2 + \cdots + X_m = \frac{m}{2} (q-1)$ is invariant under conjugation and both varieties $H_1$ and $H_2$ intersect $\pi$ at the same set $S$ of points. $S$ is the set of points in $\mathbb{Z}^m$ where $m/2$ coordinates are equal to zero and the remaining ones equal $q-1$. The facts that the points in $S$ belong to the facets of $\mathfrak{H}$,  $H_1$ is convex and $H_2$ concave on $\mathfrak{H}$ determine a geometric configuration that proves the result. Case (2) can be proved in a similar way although in this case $H_1$ and $H_2$ meet $\pi$ at the set of points where $(m-1)/2$ coordinates equal zero, $(m-1)/2$ coordinates equal $q-1$ and the remaining one is equal to $(q-1)/2$.

To finish the proof, assume we are in case (3) and consider the hyperplanes $\pi_1: X_1+ X_2 +\cdots+ X_m=\frac{m-1}{2} (q-1)+\frac{q}{2}$ and $\pi_2: X_1+ X_2 +\cdots+ X_m=\frac{m-1}{2} (q-1)+\frac{q}{2}-1$. Both hyperplanes are conjugated one of each other and the same happens with the varieties $H_i$. In addition, within $\overline{\mathfrak{H}}$, $H_1$ meets $\pi_1$ at the set of points satisfying that $(m-1)/2$ coordinates are equal to zero, $(m-1)/2$ coordinates equal $q-1$ and the remaining one is $q/2$. With respect to $H_2$ and $\pi_2$, a similar situation happens but the remaining coordinate must be $(q/2) -1$. This fact shows that although one can find nonrational points of $\overline{\mathfrak{H}}$ under $H_1$ which are not under $H_2$, this fact cannot happen with rational points because the terms of the right hand of the equations for $\pi_1$ and $\pi_2$ differ in one unit. As a consequence, Condition (*) holds and the result is proved.
\end{proof}

\subsection{Affine variety codes} \label{affine} Consider again the ring of polynomials  $\mathbb{F}_{q}[X_1, X_2, \ldots, X_m]$ and, in this case, choose $m$ positive integers $N_i$, $1 \leq i \leq m$, satisfying that $N_i$ divides $q-1$. Now the ideal $I$ defining $R= \mathbb{F}_{q}[X_1, X_2, \ldots, X_m]/I$ will be that spanned by the set of polynomials $\{X_1^{N_1}-1, X_2^{N_2}-1, \ldots, X_m^{N_m}-1\}$ and the set of evaluating points $Z(I)=\{P_j\}_{j=1}^n$. As above, we will use the morphism of vector spaces
$\mathrm{ev}: R \rightarrow \mathbb{F}_{q}^n$. Consider the cartesian product \begin{equation}\label{eq:H}\mathfrak{H} = \{0,1,\ldots,N_1-1\}\times \{0,1,\ldots,N_2-1\} \times\cdots\times\{0,1,\ldots,N_m-1\}\end{equation}  and for any nonempty subset $\Delta \subset \mathfrak{H}$, we define the {\it affine variety code} given by $\Delta$, $E_\Delta$, as the vector subspace over $\mathbb{F}_{q}$ of $\mathbb{F}_{q}^n$ spanned by the evaluation by $\mathrm{ev}$ of the classes in $R$ of the set corresponding to monomials $X^{\boldsymbol{\alpha}}  = X_1^{\alpha_1} X_2^{\alpha_2} \cdots X_m^{\alpha_m}$ such that $\boldsymbol{\alpha}=(\alpha_1,\alpha_2, \ldots, \alpha_m) \in \Delta$. Note that the length of these codes is $n=N_1 N_2 \cdots N_m$.

For a set $\Delta$ as above, we define the subset of $\mathfrak{H}$,
$\Delta^\perp  =  \mathfrak{H} \setminus \{ \hat{\boldsymbol{\alpha}} | \boldsymbol{\alpha} \in \Delta \}$,  where $\hat{\boldsymbol{\alpha}} $ denotes the element $\hat{\boldsymbol{\alpha}}=(\hat{\alpha}_1,\hat{\alpha}_2, \ldots, \hat{\alpha}_m)$ where, for $1 \leq i \leq m$, $\hat{\alpha}_i$ is $0$ whenever  $\alpha_i=0$ and $\hat{\alpha}_i =N_i - \alpha_i$ otherwise.  One can also define $\hat{\boldsymbol{\alpha}}$  as $ - \boldsymbol{\alpha} $ in $\mathbb{Z}_{N_1} \times \cdots \times \mathbb{Z}_{N_m}$.  Concerning duality, the main result is the following one that is an extension of one in \cite{maria-michael,diego}.

\begin{pro}\label{dual:affine}
The dimension of an affine variety code $E_\Delta$, defined by a set $\Delta$ as above, is the cardinality of the set $\Delta$. Moreover, the dual code $E_\Delta^\perp$ of $E_\Delta$ is the affine variety code $E_{\Delta^\perp}$.
\end{pro}
\begin{proof}
Let $\xi_i \in \mathbb{F}_q$ with order $N_i$, for $i=1, 2, \ldots, m$, whose existence is guaranteed by the fact that $N_i | q-1$. So, $\langle \xi_i \rangle = \{ \xi_i^0 , \xi_i^1 , \ldots , \xi_i^{N-1} \} = Z(X_i^{N_i} -1)$.

Let $\boldsymbol{a}, \boldsymbol{b} \in \mathfrak{H}$, by the distributive property,  $\mathrm{ev} (X^{\boldsymbol{a}}) \cdot \mathrm{ev} (X^{\boldsymbol{b}}) $  is equal to $$\left(\sum_{\gamma_1 \in \langle \xi_1 \rangle} \gamma_1^{a_1 + b_1} \right) \left(\sum_{\gamma_2 \in \langle \xi_2 \rangle} \gamma_2^{a_2 + b_2} \right) \cdots \left( \sum_{\gamma_m \in \langle \xi_m \rangle} \gamma_m^{a_m + b_m}   \right).$$

If $a_i + b_i =0$ in $\mathbb{Z}_{N_i}$ for every $i \in \{1 , 2, \ldots, m\}$, then $ \mathrm{ev} (X^{\boldsymbol{a}}) \cdot \mathrm{ev} (X^{\boldsymbol{b}}) \neq 0 $ because $$\sum_{\gamma_i \in \langle \xi_i \rangle} \gamma_i^{a_i + b_i} = \sum_{\gamma_i \in \langle \xi_i \rangle} \gamma_i^0 = N_i \neq 0 ~(\text{in}~\mathbb{F}_q).$$ However, if $a_i + b_i = c \neq 0$ in $\mathbb{Z}_{N_i}$ for some $i$, then $ \mathrm{ev} (X^{\boldsymbol{a}}) \cdot \mathrm{ev} (X^{\boldsymbol{b}}) = 0 $  because $$\sum_{\gamma_i \in \langle \xi_i \rangle} \gamma_i^{a_i + b_i} = \sum_{j=0}^{N_i-1} (\xi_i^j)^{c} = \sum_{j=0}^{N_i-1} (\xi_i^c)^{j} = \frac{1-(\xi_i^c)^{N_i}}{1 - \xi_i^c} = 0;$$note that $\xi_i^c \neq 1$ since $c \neq 0 \in  \mathbb{Z}_{N_i}$.

Then $\mathrm{ev} (X^{\boldsymbol{a}}) \cdot \mathrm{ev} (X^{\boldsymbol{b}})  = 0$ for $\boldsymbol{a} \in \Delta$, $\boldsymbol{b} \in \Delta^\perp$ since $\boldsymbol{a} + \boldsymbol{b} \neq \boldsymbol{0}$ in $\mathbb{Z}_{N_1} \times \mathbb{Z}_{N_2} \times \cdots \times \mathbb{Z}_{N_m}$.  On account of the dimension of $E_\Delta$ and $E_{\Delta^\perp}$ and the linearity of the codes the result holds.

\end{proof}

\section{Subfield-subcodes}
\label{subfield}
Section \ref{families} provides codes, including their derived matrix-product ones, suitable to get stabilizer codes via the CSS code construction and the Steane's enlargement procedure. However, stabilizer codes with better parameters can be obtained by considering subfield-subcodes. Next, we are going to give some details with respect to their  dimensions.

Recall that $q=p^r$, assume $r>2$ and pick a positive integer $s<r$ such that $s$ divides $r$. Consider the trace type maps: $\mathrm{tr}_r^s: \mathbb{F}_{p^r} \rightarrow \mathbb{F}_{p^s}$ defined as $\mathrm{tr}_r^s(x)= x + x^{p^s} + \cdots + x^{p^{s(\frac{r}{s}-1)}}$; $\mathbf{tr}: \mathbb{F}_{p^{r}}^n \rightarrow \mathbb{F}_{p^{s}}^n$, which works by applying  $\mathrm{tr}_r^s$ componentwise and, for the different rings $R$ defined in Section \ref{families}, $\mathcal{T}: R\rightarrow R$, $\mathcal{T}(f) = f + f^{p^s} + \cdots + f^{p^{s(\frac{r}{s}-1)}}$. We must add that we consider $f \in R$ given by a linear combination of monomials with exponents in $\mathfrak{H}$, $\mathfrak{H}$ being the hypercube $ (\{0,1,\ldots,q-1 \})^m$  in Subsections \ref{reed-muller} and \ref{hyperbolic},  and $\mathfrak{H}$ as defined in (\ref{eq:H}), in Subsection \ref{affine}.  In the rest of this section, we will set $N_i=q$, $1 \leq i \leq m$, when we are working with either Reed-Muller or hyperbolic codes. Otherwise, $N_i$ will be the corresponding values for affine variety codes.

\begin{rem}{\rm
To define the codes in the previous section, we have considered the algebra $R= \mathbb{F}_{q}[X_1,  \ldots, X_m]/I$, where $I$ is spanned by the set of polynomials $\{X_1^{N_1}-X_1, \ldots, X_m^{N_m}-X_m\}$  for Reed-Muller and hyperbolic codes, and by the set  $\{X_1^{N_1}-1,   \ldots, X_m^{N_m}-1\}$ for affine variety codes. Therefore, the algebra $R$ is slightly different for affine variety codes in this work, the only difference residing in the fact that we are only evaluating at points with nonzero coordinates. Although the literature usually considers affine variety codes using the first ideal,  we have decided to consider the second ideal in order to compare some of our codes with the ones in \cite{lag3,lag1}, whose length is a power of $q$ minus one. }
\end{rem}

For each index $i$ as above, set $\mathbb{Z}_{N_i}$ the quotient ring $\mathbb{Z} / N_i \mathbb{Z}$.  A subset $\mathfrak{I}$ of the cartesian product $\mathbb{Z}_{N_1}\times \mathbb{Z}_{N_2} \times \cdots\times\mathbb{Z}_{N_m}$ is a {\it cyclotomic set} if  it satisfies
$
\mathfrak{I}= \{p \cdot \boldsymbol{\alpha} \;| \; \boldsymbol{\alpha}\in \mathfrak{I}\}
$,
where $p \cdot \boldsymbol{\alpha} = (p \alpha_1, p \alpha_2, \ldots, p \alpha_m)$. Moreover, a cyclotomic set $\mathfrak{I}$ is called {\it minimal} (for the exponent $s$ above introduced) whenever every element in $\mathfrak{I}$ can be expressed as  $p^{s j } \cdot \boldsymbol{\alpha}$ for some fixed $\boldsymbol{\alpha} \in \mathfrak{I}$ and some nonnegative integer $j$. Fixing a representant $\mathbf{a} \in \mathfrak{I}$ for each minimal cyclotomic set, one gets a set of representatives $\mathcal{A}$. Then, $ \mathfrak{I} = \mathfrak{I}_\mathbf{a}$ for some $\mathbf{a} \in \mathcal{A}$. The family of minimal cyclotomic sets, with respect to $s$, will be $\{ \mathfrak{I}_\mathbf{a}\}_{\mathbf{a} \in \mathcal{A}}$ and we will denote $i_\mathbf{a} : = \mathrm{card}(\mathfrak{I}_\mathbf{a})$. In addition, $r$ is a multiple of $i_\mathbf{a}$ and, setting $\mathbf{a} = (a_1,a_2, \ldots, a_m)$, the congruence $a_i \cdot p^{s i_\mathbf{a}} \equiv a_i \mod N_i$ holds.

The main advantage of considering cyclotomic sets is that any element $f \in R$ can be uniquely decomposed in the form $
f = \sum_{\mathbf{a} \in \mathcal{A}} f_{\mathbf{a}}$,
where $ f_{\mathbf{a}}$ are classes of polynomials in $R$ whose support (that of its canonical representative), $\mathrm{supp} (f_\mathbf{a})$, is included in $\mathfrak{I}_\mathbf{a}$. Furthermore, it holds
\[
\mathrm{supp} (\mathcal{T}\left(f_\mathbf{a})\right) \subset \mathfrak{I}_\mathbf{a}.
\]

Our aim in this section consists of describing  subfield-subcodes of the families of codes introduced in Section \ref{families}.
We will need the concept of {\it element $f \in R$  evaluating to $\mathbb{F}_{p^s}$}. This means that  $f(\boldsymbol{\alpha}) \in \mathbb{F}_{p^s}$ for all $\boldsymbol{\alpha}  \in Z(I)$. This happens if and only if $f = \mathcal{T}(g)$ for some $g \in R$. Then, we can state the following result \cite{galindo-hernando}.

\begin{teo}
\label{ga-he}
Let $\beta_\mathbf{a}$ be a primitive element of the finite field $\mathbb{F}_{p^{si_\mathbf{a}}}$ and set $\mathcal{T}_\mathbf{a} : R \rightarrow R$ the mapping defined as $\mathcal{T}_\mathbf{a}(f) = f + f^{p^s} + \cdots + f^{p^{s(i_\mathbf{a} -1)}}$. Then, a basis of the vector space of elements in $R$ evaluating to  $\mathbb{F}_{p^s}$ is
\[
 \bigcup_{\mathbf{a} \in \mathcal{A} } \left\{ \mathcal{T}_{\mathbf{a}} (\beta_\mathbf{a}^{l} X^\mathbf{a}) \; | \;
0 \leq l \leq i_\mathbf{a}-1  \right\}.
\]
\end{teo}

Denote  by $\Delta$ any of the sets generating (by applying $\mathrm{ev}$ to the monomials that represent) any of the codes described in Section \ref{families}, which we set $E_\Delta$. $\Delta$ is a subset of the hypercube $\mathfrak{H}$ for affine variety codes. For Reed-Muller codes, $\Delta$ will be  the exponents set of the monomials in $R$ of total degree less than or equal to certain positive integer $r$. Finally, when we consider a hyperbolic code, $\Delta$ will be  the set of exponents appearing in the monomials in (\ref{HHyp}) for some value $t$ as above mentioned. Consider the set $E^\sigma_\Delta = E_\Delta \cap \mathbb{F}_{p^s}^n$. $E^\sigma_\Delta$ is defined by the traces $ \mathcal{T} (g)$ of elements $g \in R$ such that $ \mathcal{T} (g)$  is in the vector space generated by monomials with exponents in $\Delta$. As a consequence, one gets

\begin{teo}
The vector space $E^\sigma_\Delta$ is generated by the images under the evaluation map $\mathrm{ev}$ of the elements in $R$
\[
\bigcup_{\mathbf{a} \in \mathcal{A}| \mathfrak{I}_\mathbf{a} \subset \Delta}  \left\{ \mathcal{T}_{\mathbf{a}} (\beta_\mathbf{a}^{l} X^\mathbf{a}) \; | \; 0 \leq l \leq i_\mathbf{a}-1  \right\}.
\]
\end{teo}

With respect to the dual code of $E^\sigma_\Delta$, one can consider the following diagram:
\[
\begin{CD}
E_\Delta @>\mathrm{duality}>> E_\Delta^\perp\\
@VVV @V\mathbf{tr}VV \\
E^\sigma_\Delta = E_\Delta \cap \mathbb{F}_{p^s}^n @>>\mathrm{duality}> (E^\sigma_\Delta)^\perp = (E^\perp_\Delta)^\sigma
\end{CD}
\]
where we notice that the equality at the bottom right holds by Delsarte Theorem \cite{delsarte}.

When we are dealing with affine variety codes $E_\Delta$, we have defined in Subsection \ref{affine} the set $\Delta^\perp$ attached with $\Delta$ and defined the corresponding dual code. Analogously, for Reed-Muller codes, defined by  the set $\Delta$, corresponding to monomials in $R$ of degree less than or equal to $r$, we can define $\Delta^\perp$ as the set of exponents of monomials in $R$ of degree less than or equal to $m(q+1) -(r+1)$. Finally, for the case of hyperbolic codes Hyp$(t,m)$, the set $\Delta^\perp$ is showed in (\ref{DosDos}). Setting $C_\Delta = E^\perp_\Delta$ and $C_\Delta^\sigma =  C_\Delta \cap \mathbb{F}_{p^s}^n$, the above diagram proves that $C_\Delta^\sigma = (E_\Delta^\sigma)^\perp$ and thus $C_\Delta^\sigma$ is the vector space generated by $\mathbf{tr} \left( \mathrm{ev}( \Delta^\perp)\right)$, that is the vector space generated by $\mathrm{ev}\left( \mathcal{T}( \Delta^\perp)\right)$,
where $\Delta^\perp$ is defined as above. As a consequence, one gets the following result:
\begin{teo}
\label{perp}
Let $\Delta$ be the defining set of a code as above. Consider its corresponding set $\Delta^\perp$. With the above notations, the dual code $C_\Delta^\sigma$ of the code $E_\Delta^\sigma$ is generated by those vectors in $\mathbb{F}_{p^s}^n$ obtained by applying the map $\mathrm{ev}$ to the following set of elements in $R$
\[
\bigcup_{\mathbf{a} \in \mathcal{A} | \mathfrak{I}_\mathbf{a} \cap \Delta^\perp \neq \emptyset}  \left\{ \mathcal{T}_{\mathbf{a}} (\beta_\mathbf{a}^{l} X^\mathbf{a}) \; | \; 0 \leq l \leq i_\mathbf{a}-1  \right\}.
\]
\end{teo}

Finally we state the next result which extends to Reed-Muller and hyperbolic codes Theorems 5 and 6 in \cite{galindo-hernando}.

\begin{teo}
\label{eseldoce}
Let $\Delta$ be an evaluating set as in Theorem \ref{perp} providing a Reed-Muller, a hyperbolic or an affine variety code. Consider the subfield-subcode $E_\Delta^\sigma$ and its dual one $C_\Delta^\sigma$. Then
\begin{enumerate}
  \item The dimension of the code $C_\Delta^\sigma$ can be computed as
  \[ \dim (C_\Delta^\sigma)= \sum_{\mathbf{a} \in \mathcal{A}| \mathfrak{I}_\mathbf{a} \cap \Delta^\perp \neq \emptyset} i_\mathbf{a}.
\]
  \item The inclusion $E_\Delta^\sigma \subset C_\Delta^\sigma$ holds if, and only if,  $\mathfrak{I}_\mathbf{a} \cap \Delta^\perp \neq \emptyset$ whenever $\mathfrak{I}_\mathbf{a} \subset \Delta$, which in case of affine variety codes can be expressed as $\left\{ \hat{\boldsymbol{\alpha}} \; | \; \boldsymbol{\alpha} \in \mathfrak{I}_\mathbf{a} \right\} \not \subset \Delta$
      whenever $\mathfrak{I}_\mathbf{a} \subset \Delta$.
\end{enumerate}
\end{teo}

\section{Quantum stabilizer codes}
\label{quantumm}
This section is devoted to state results concerning quantum stabilizer codes by using results in previous sections. In this paper,  using only Euclidean inner product, we will get good stabilizer codes from the above studied codes and their matrix-product codes. A nice way to do it employs orthogonal matrices over finite fields. By using a computer, it is not difficult to obtain such matrices for fields of small cardinality. We are especially interested in this situation because, in most cases, we will use subfield-subcodes. For larger fields, one can use orthogonal circulant matrices because, according to \cite{junbe} (see also \cite{mac}), one may check whether the matrix is orthogonal by a condition on the polynomials determined by the first row and column of the matrix.


Now we are ready to state our main results. $A_q$ will be an orthogonal $s \times s$ matrix over a finite field $\mathbb{F}_q$ with attached code distances $\{\delta_i\}_{1 \leq i \leq s}$ as defined in Section \ref{matrix}.

\begin{teo}
\label{UNO}
Let $\{E_i\}_{i=1}^s$ (respectively, $E_1 \subset E_2 \subset \cdots \subset E_s$) be a sequence (respectively,  a nested sequence) of  codes over a finite field $\mathbb{F}_q$ of one of the following types:

a) The codes $E_i =RM (r_i,m)$, $m >0$, $1 \leq i \leq s$, are Reed-Muller codes attached with a sequence of positive integers $\{r_i\}_{i=1}^s$ (respectively, $r_1 < r_2 < \cdots < r_s$) satisfying $2 r_i + 1 < m (q-1)$ for all $i$ (respectively, $2 r_s + 1 < m (q-1)$).

b) Each code $E_i$, $1 \leq i \leq s$, is spanned by the vectors of $\mathbb{F}_q^n$ obtained applying $\mathrm{ev}$ to the set of monomials
      $
      \left\{ X^{\boldsymbol{\alpha}} \;|\; 0 \leq \alpha_j < q, 1 \leq j \leq m, \prod_{j=1}^m (\alpha_j + 1) < q^m -t_i \right\}
      $
such that, for all $i$, $t_i$ is a positive integer as in Subsection \ref{hyperbolic} and satisfies Theorem \ref{te: ConditionSelfOrthogonal} (respectively, $t_i$ is a sequence of positive integer as in Subsection \ref{hyperbolic} such that $q^m > t_1 > t_2 > \cdots >t_s$  and $t_s$ satisfies Theorem \ref{te: ConditionSelfOrthogonal}).

c) $E_i= E_{\Delta_i}$ are affine variety codes such that  for all $i$ $ \hat{\boldsymbol{\alpha}} \not \in \Delta_i$ (respectively, $ \hat{\boldsymbol{\alpha}} \not \in \Delta_s$ and $\Delta_j \subset \Delta_{j+1}$, $1 \leq j \leq s-1$) whenever $ \boldsymbol{\alpha}\in \Delta_i$ (respectively, $ \boldsymbol{\alpha}\in \Delta_s$).\\

Consider the sequence of dual codes $\{C_i\}_{i=1}^s$(respectively, $C_1 \supset C_2 \supset \cdots \supset C_s$), where $C_i  =  E_i^\perp$. Then,

(1) The matrix-product code $\mathfrak{C}  =  [C_1,C_2, \ldots, C_s]\cdot A_q$
  is an $[\mathfrak{n},\mathfrak{k},\mathfrak{d}]$-code over $\mathbb{F}_q$ where $ \mathfrak{n} =n s$, $ \mathfrak{k}= \sum_{i=1}^s k_i$ and $\mathfrak{d} \geq  \min_{1 \leq i \leq s} \{\delta_i d_i\} \; (\mathrm{respectively}\;   \mathfrak{d} = \min_{1 \leq i \leq s} \{\delta_i d_i\}  )\;$. Moreover according the above cases, the following statements hold.\\
In case a), the dimensions $k_i$ satisfy the equality (\ref{DIM}) with $r_i$ instead of $r$ and, for all $i$, $d_i= (b_i +1) q^{a_i}$, where $r_i +1 = a_i (q-1) + b_i$ obtained by Euclidean division.\\
In case b), the dimensions $k_i$ satisfy the equality in (3) of Theorem \ref{elcinco} with $t_i$ instead of $t$ and, for all $i$, $d_i = q^m - t_i$.\\
In case c), it happens  $k_i = \mathrm{card} (\Delta_i)$ and $d_i = \dim C_i$.

(2) $\mathfrak{C}$ provides a stabilizer code with parameters $[[\mathfrak{n},\mathfrak{K},\geq \mathfrak{d}]]_q$, where $\mathfrak{K}= 2 \mathfrak{k} - \mathfrak{n}$.
\end{teo}
\begin{proof}
Theorem \ref{fer1} together with results in Subsection \ref{reed-muller} prove Statement (1) a). The same happens with Statement (1) b) and (1) c) if one uses Subsections \ref{hyperbolic} and \ref{affine}, respectively. Corollaries \ref{coro2} and \ref{bueno} prove our Statement (2).
\end{proof}

For subfield-subcodes, we get

\begin{teo}
\label{DOS}
Let $\{E_i\}_{i=1}^s$ (respectively, $E_1 \subset E_2 \subset \cdots \subset E_s$) be a family of codes (respectively, a nested sequence of  codes) over a finite field $\mathbb{F}_q$ defined as in Theorem \ref{UNO}, $q=p^r$ and consider the finite field $\mathbb{F}_{p^s}$, where $s$ divides $r$. Set $\Delta_i$, $ 1 \leq i \leq s$, the subsets of the ring $R$ whose evaluation provides $E_i$ and assume the following property: $\mathfrak{I}_\mathbf{a} \cap \Delta^\perp \neq \emptyset$ whenever $\mathfrak{I}_\mathbf{a} \subset \Delta$, where $\Delta$ and the minimal cyclotomic subsets $\mathfrak{I}_\mathbf{a}$ are as described in Section \ref{subfield}. Then, the matrix-product code $\mathfrak{C}^\sigma  =  [C_1^\sigma,C_2^\sigma, \ldots, C_s^\sigma] A_{p^s}$ is an $[\mathfrak{n},\mathfrak{k},\mathfrak{d}]$-code over $\mathbb{F}_{p^s}$, where $ \mathfrak{n} =n s$, $ \mathfrak{k}= \sum_{i=1}^s k_i$, where $
k_i = \sum_{\mathbf{a} \in \mathcal{A} | \mathfrak{I}_\mathbf{a} \cap \Delta^\perp \neq \emptyset} i_\mathbf{a}$  and $\mathfrak{d} \geq \; (\mathrm{respectively}, \; =) \; \min_{1 \leq i \leq s} \{\delta_i d_i\}$, the distances $d_i$ being as in Theorem \ref{UNO}.

Finally, $\mathfrak{C}^\sigma$ provides a stabilizer code with parameters $[[\mathfrak{n},\mathfrak{K},\geq \mathfrak{d}]]_{p^s}$, where $\mathfrak{K}= 2 \mathfrak{k} - \mathfrak{n}$.
\end{teo}
\begin{proof}
It follows from Theorems \ref{eseldoce} and \ref{UNO} and Corollaries \ref{coro2} and \ref{bueno}.
\end{proof}

\begin{rem} \label{U}
{\rm
The families of codes considered in this paper allow us to construct sequences of nested codes $C_1^\perp \subset \cdots  \subset C_s^\perp  \subset C_s \subset \cdots \subset C_1$. These sequences contain either  codes as in Section \ref{families} or subfield-subcodes or even matrix-product codes coming from the previous mentioned codes. Results in Section \ref{families} and Theorems \ref{UNO} and \ref{DOS} give parameters $[n,k_i, d_i]$ for the above codes and one can get stabilizer codes by using Corollary \ref{ham}. Indeed, consider two suitable subindices $1 \leq i < j \leq s$, then $C_j^\perp \subset C_j \subset C_i$ and, therefore, there exists an $[[n, k_i + k_j - n, \geq \min \{ d_j, \lceil \frac{q+1}{q}\rceil d_i\}]]_q$ stabilizer code.
}
\end{rem}

We devote the rest of the paper to provide new stabilizer codes over different base fields by using the results above stated.

\section{Quantum stabilizer codes over $\mathbb{F}_2$}
\label{efe2}

Along this section we will assume that our field is  $\mathbb{F}_2$. Firstly, we are going to show several quantum binary codes that improve the best known parameters given in \cite{codet}. Afterwards we will show some good stabilizer obtained with matrix-product codes coming from three constituent codes.

\subsection{} We get codes improving \cite{codet} by applying Corollary \ref{ham} to suitable subfield-subcodes of certain affine variety codes and also from some subcodes and extended codes of them. Indeed, with ideas and notations as in Section \ref{subfield}, set $p=2, r=7, s=1$ and $N_1=127$. The following table shows parameters and defining sets $\Delta$ of stabilizer codes obtained with the CSS code construction  of subfield-subcodes of the mentioned affine variety codes. Notice that, from Corollary \ref{bueno}, it is straightforward to get the parameters of the originally used linear codes.
\vspace{1mm}
\begin{center}
\begin{tabular}{||c|c|c|c|c||}
  \hline \hline
  Code  & $n$ & $k$ & $d \geq$  & Defining set $\Delta$ \\
  \hline \hline
  $C_1$  & 127 & 85 & 7 &   $ \{46,92,57,114,101,75,23,110,93,59,118,$\\
     & &  &  &
   $ 109,91,55,
    38,76,25,50,100,73,19
\}$ \\
  \hline
    $C_2$ & 127 & 57 & 11 &   $\{ 46,92,57,114,101,75,23,
    110,93,59,118,109,91,55,
    38,76,25,50,$\\
     & &  &  &
   $ 100,73,19,
 42,84,41,82,37,74,21,
58,116,105,83,39,78,29
\}$ \\
 \hline
    $C_3$ & 127 & 71 & 9 &   $\{ 42,84,41,82,37,74,21,
6,12,24,48,96,65,3,
18, $\\
     & &  &  &
   $ 36,72,17,34,68,9,
30,60,120,113,99,71,15
\}$ \\
     \hline
   $C_4$ & 127 & 43 & 13 &  $\{ 42,84,41,82,37,74,21,
    54,108,89,51,102,77,$\\
     & &  &  &
   $ 27,
    6,12,24,48,96,65,3,
    58,116,105,$\\
     & &  &  &
   $ 83,39,78,29,
    18,36,72,17,34,68,9,
    30,60,120,113,99,71,15
\}$ \\
\hline
 \hline
\end{tabular}
\captionof{table}{Stabilizer affine variety codes over $\mathbb{F}_2$}
\end{center}
\vspace{2mm}

Steane's enlargement (in short, SE), Corollary \ref{ham}, applied to the codes $C_2$ and $C_1$ (respectively, $C_4$ and $C_3$) provides stabilizer codes $C_5$ and $C_6$. Next table shows parameters of these codes and some of their modifications. All these codes improve the best known records, which can be seen in \cite{codet} and have the same values $n$ and $k$ but distance one unit less.
\vspace{3mm}
\begin{center}
\begin{tabular}{||c|c|c|c||}
  \hline \hline
  Code  & $n$ & $k$ & $d \geq$  \\
  \hline \hline
  $C_5$ =  $\mathrm{SE} (C_2,C_1)$& 127 & 71 & 11   \\
  \hline
  Extended code $(C_5)$ & 128 & 71 & 11   \\
 \hline
 Subcode $(C_5, 70)$ & 127 & 70 & 11  \\
 \hline
  Subcode $(C_5, 69)$ & 127 & 69 & 11  \\
 \hline
 $C_6$ = $\mathrm{SE}(C_4,C_3)$ & 127 & 57 & 13   \\
  \hline
  Extended code$(C_6)$ & 128 & 57 & 13   \\
  \hline
   Subcode $(C_6, 56)$ & 127 & 56 & 13  \\
 \hline
  \hline
\end{tabular}
\captionof{table}{Best known stabilizer codes over $\mathbb{F}_2$}\label{ta:best}
\end{center}

\subsection{} There is no non-singular by columns orthogonal matrix of size $3\times 3$ over  $\mathbb{F}_2$, however matrix-product codes suitable  by providing quantum codes with $s=3$  can be obtained by using the following matrix over $\mathbb{F}_2$, whose  transpose inverse is also displayed.
\begin{equation}
\label{matriz2}
A=\begin{pmatrix}
1 & 0 & 1 \\
1 & 1 & 0  \\
1 & 1 & 1
\end{pmatrix}, \; \;
(A^{-1})^t=\begin{pmatrix}
1 & 1 & 0 \\
1 & 0 & 1  \\
1 & 1 & 1
\end{pmatrix}.
\end{equation}

\begin{teo}\label{te: QMPCsobreF2}
Let $C_1$ and $C_2$ be linear codes over $\mathbb{F}_2$ with parameters $[n,k_1,d_1]$ and $[n,k_2,d_2]$  respectively, and  such that  $C_1\supset C_1^{\perp}$ and $C_2\supset C_2^{\perp}$.  If $A$ is the matrix showed in (\ref{matriz2}), then the following inclusion involving matrix-product codes holds:
\[
[C_1,C_1,C_2]\cdot A\supset ([C_1,C_1,C_2]\cdot A)^{\perp}.
\]
Moreover, the previous constructed code yields a stabilizer code with parameters
\[
[[3n,2(2k_1+k_2)-3n,\ge \min\{2d_1,d_2\}]]_2.	
\]
\end{teo}
\begin{proof}
The definition of matrix-product code and Theorem \ref{norton} show that a generic codeword of $([C_1,C_1,C_2]\cdot A)^{\perp}$ is of the form $(\mathbf{c}_1+\mathbf{c}_1'+\mathbf{c}_2,\mathbf{c}_1+\mathbf{c}_2,\mathbf{c}_1'+\mathbf{c}_2)$,
where $\mathbf{c}_1,\mathbf{c}_1'$ are generic elements in the code $C_1^{\perp}$ and $\mathbf{c}_2$ in $C_2^{\perp}$. Since $C_1\supset C_1^{\perp}$ and $C_2\supset C_2^{\perp}$, we have that
$\mathbf{c}_1\in C_1$ and $\mathbf{c}_2\in C_2$. Switching the roles of $\mathbf{c}_1$ and $\mathbf{c}_1'$, we conclude that $(\mathbf{c}_1+\mathbf{c}_1'+\mathbf{c}_2,\mathbf{c}_1+\mathbf{c}_2,\mathbf{c}_1'+\mathbf{c}_2)$ is also in  $[C_1,C_1,C_2]\cdot A$, which proves our first statement.

By Theorem \ref{fer1}, the parameters of the matrix-product code $[C_1,C_1,C_2]\cdot A$ are $[3n,2k_1+k_2,\ge \min\{2d_1,d_2\}]$. Hence applying Corollary \ref{bueno}, it is obtained a stabilizer code with parameters $[[3n,2(2k_1+k_2)-3n,\ge \min\{2d_1,d_2\}]]_2$.
\end{proof}

To finish this section, we are going to give some examples of stabilizer codes given by matrix-product codes as showed in Theorem \ref{te: QMPCsobreF2}. We use the same notation as in Subsections \ref{reed-muller} and \ref{hyperbolic}  and consider two cases: $m=4$ and $m=6$. In each case, we show a table containing parameters of stabilizer codes obtained from Reed-Muller or hyperbolic codes by using Corollary \ref{bueno}. Later, we present a second table of codes which are obtained with the matrix $A$ in (\ref{matriz2}) and following Theorem \ref{te: QMPCsobreF2}; in the case denoted by SE$(D_3,D_1)$, we use the Steane's enlargement procedure of the codes $D_3$ and $D_1$. With respect to the case $m=4$, we set: \vspace{5mm}
\begin{center}
\begin{tabular}{||c|c|c|c||}
  \hline \hline
  Code  & $n$ & $k$ & $d \geq$  \\
  \hline \hline
  $C_1$  & 16 & 16 & 1   \\
  \hline
  $C_2$  & 16 & 14 & 2   \\
 \hline
    $C_3$ & 16 & 6 & 4  \\
  \hline
  \hline
\end{tabular}
\captionof{table}{Stabilizer codes over $\mathbb{F}_2$ by using Corollary \ref{bueno}, $m=4$}
\end{center}

As mentioned, it is not difficult to get the parameters of the original codes; for instance the stabilizer code $[[16,14,2]]_2$ comes from a (Reed-Muller) code $C_1$ over $\mathbb{F}_2$ with parameters $[16,15,2]$. Applying Theorem \ref{te: QMPCsobreF2}, we get stabilizer codes over $\mathbb{F}_2$ from matrix-product codes obtained with the previous codes $C_i$. Their parameters are:
\vspace{2mm}
\begin{center}
\begin{tabular}{||c|c||}
  \hline \hline
  Matrix-Product Code & Quantum Parameters  \\
  \hline \hline
  $D_1 := [C_1,C_1,C_2] \cdot A$ & $[[48,46, \geq 2]]_2$     \\
  \hline
   $D_2 : = [C_2,C_2,C_2] \cdot A$ & $[[48,42, \geq 2]]_2$     \\
  \hline
  $D_3 := [C_2,C_2,C_3] \cdot A$ & $[[48,34, \geq 4]]_2$     \\
 \hline
  $ \mathrm{SE} (D_3,D_1)$ & $[[48,40, \geq 3]]_2$     \\
  \hline
  \hline
\end{tabular}
\captionof{table}{Stabilizer codes of length $48$ over $\mathbb{F}_2$ by using Theorem \ref{te: QMPCsobreF2}}
\label{tablap14}
\end{center}

Codes with these parameters are known, however this is a sample that good codes can be obtained with our techniques because our quantum code $[[48,34,4]]_2$ is as  good as the best known quantum code with that length and dimension \cite{codet}. In addition, the parameters of the remaining codes in Table \ref{tablap14} cannot be improved.

Finally, with respect to $m=6$, we  give the following two tables: \vspace{2mm}
\begin{center}
\begin{tabular}{||c|c|c|c||}
  \hline \hline
  Code  & $n$ & $k$ & $d \geq$  \\
  \hline \hline
  $C_4$  & 64 & 64 & 1   \\
  \hline
  $C_5$  & 64 & 62 & 2   \\
  \hline
    $C_6$ & 64 & 50 & 4   \\
    \hline
    $C_7$ & 64 & 20 & 8   \\
  \hline
  \hline
\end{tabular}
\captionof{table}{Stabilizer codes over $\mathbb{F}_2$ by using Corollary \ref{bueno}, $m=6$}
\label{tabla1}
\end{center}
\vspace{2mm}

\begin{center}
\begin{tabular}{||c|c||}
  \hline \hline
  Matrix-Product Code & Quantum Parameters  \\
  \hline \hline
   $[C_{4},C_{4},C_{5}] \cdot A$ & $[[192,190, \geq 2]]_2$     \\
 \hline
  $[C_{5},C_{5},C_{6}] \cdot A$ & $[[192,174, \geq 4]]_2$     \\
   \hline
  $[C_{6},C_{6},C_{7}] \cdot A$ & $[[192,120, \geq 8]]_2$     \\
  \hline
  \hline
\end{tabular}
\captionof{table}{Stabilizer codes over $\mathbb{F}_2$ by using Theorem \ref{te: QMPCsobreF2}}
\label{tablap15}
\end{center}
\vspace{3mm}
Note that the first two codes in Table \ref{tablap15} exceed the Gilbert-Varshamov bounds \cite{eck,mat,feng}, \cite[Lemma 31]{kkk}.

\section{Quantum stabilizer codes over $\mathbb{F}_3$}
\label{efe3}

As in the previous section, we desire to give parameters for some stabilizer codes over the field $\mathbb{F}_3$. We would also like use matrix-product codes, however, again in this case, there is no  nonsingular by columns orthogonal matrix over $\mathbb{F}_3$ of size either $2\times 2$ or  $3\times 3$. A way to avoid this problem consists of using matrix-product codes with only two constituent  codes. To do it, we consider the  following matrix and its  transpose inverse.
\begin{equation}
\label{matriz4}
A=\begin{pmatrix}
1 & 1 \\
2 & 1   \\
\end{pmatrix}, \;\; (A^{-1})^t=\begin{pmatrix}
2 & 2  \\
1 & 2  \\
\end{pmatrix}.
\end{equation}

\begin{teo}\label{te: QMPCsobreF3}
Let $C_1,C_2$ be two linear codes over $\mathbb{F}_3$ with parameters $[n,k_1,d_1]$ and $[n,k_2,d_2]$ respectively, and such that  $C_1 \supset C_1^{\perp}$ and $C_2 \supset C_2^{\perp}$. If $A$ is the matrix given in (\ref{matriz4}), then the following codes inclusion happens
\[
[C_1,C_2]\cdot A\supset ([C_1,C_2]\cdot A)^{\perp}.
\]
Moreover, the above given matrix-product code yields a stabilizer quantum code with parameters
\[
[[2n,2(k_1+k_2-n),\ge \min\{2d_1,d_2\}]]_{3}.
\]
\end{teo}
\begin{proof}
A generic codeword of $([C_1,C_2]\cdot A)^{\perp}$ is of the form
$(2\mathbf{c}_1+\mathbf{c}_2,2\mathbf{c}_1+2\mathbf{c}_2)$,
where $\mathbf{c}_1\in C_1^{\perp}$ and $\mathbf{c}_2\in C_2^{\perp}$. Taking into account that multiplication by $2$ gives an isomorphism of the field $\mathbb{F}_3$ and  $C_1\supset C_1^{\perp}$ and $C_2 \supset C_2^{\perp}$, we have that
$(2\mathbf{c}_1+\mathbf{c}_2,2\mathbf{c}_1+2\mathbf{c}_2)\in [C_1,C_2]\cdot A $ because it corresponds to the words in $[C_1,C_2]\cdot A $ given, generically, by the elements $2\mathbf{c}_1 \in C_1$ and $2\mathbf{c}_2 \in C_2$.

The same reasoning as in Theorem \ref{te: QMPCsobreF2} proves that $[C_1,C_2]\cdot A$
 is a $[2n,k_1+k_2,\ge \min\{2d_1,d_2\}]$-code over $\mathbb{F}_3$ and yields a stabilizer code with parameters $[[2n,2(k_1+k_2-n),\ge \min\{2d_1,d_2\}]]_3$.
\end{proof}

As an example, if one considers suitable Reed-Muller codes of length 9 (respectively, hyperbolic codes of length 27) and applies Theorem \ref{te: QMPCsobreF3} and Corollary \ref{ham}, an $[[18,13,3]]_3$ (respectively, $[[54,48,3]]_3$) stabilizer code is obtained. Both of them exceed the Gilbert-Varshamov bounds \cite{eck,mat,feng}, \cite[Lemma 31]{kkk}, where we consider the natural extension to $q \neq 2$ of the bounds in \cite{eck,mat}. Note that the Gilbert-Varshamov bounds in \cite{feng}, \cite[Lemma 31]{kkk} assume that $n \equiv k~( \mathrm{mod}~ 2)$. In this paper, we say that an $[[n,k,d]]$ stabilizer code, $d \geq 2$, such that $n \not \equiv k~( \mathrm{mod}~ 2)$ exceeds these Gilbert-Varshamov bounds when the parameters $[[n,k-1,d]]$ do that. Next, we provide parameters of stabilizer codes, $C_1$ and $C_2$, coming from subfield-subcodes of Reed-Muller or hyperbolic codes over $\mathbb{F}_3$. With notations as in Sections \ref{families} and \ref{subfield}, $p=3, r=2, s=1$ and $m=2$. Larger codes can be found with Theorem \ref{te: QMPCsobreF3} as can be seen in Table \ref{tablap17}. Note that the forthcoming code $D$ can be obtained with Steane's enlargement of certain matrix-product code of Reed-Muller codes. \vspace{3mm}

\begin{center}
\begin{tabular}{||c|c|c|c|c||}
  \hline \hline
  Code  & $n$ & $k$ & $d \geq$ &  Defining set $\Delta$ \\
  \hline \hline
  $C_1$  & 81 & 79 & 2 &    $ \{(0,0)\}$ \\
 \hline
    $C_2$ & 81 & 67 & 4 &  $\{  (0, 0),(0, 7),(0, 5),( 3, 0 ),( 9, 0 ),(4,0),(0,4)$\\
 \hline
\end{tabular}
\captionof{table}{Stabilizer subfield-subcodes of Reed-Muller or hyperbolic codes over $\mathbb{F}_3$}
\end{center}

\begin{center}
\begin{tabular}{||c|c||}
  \hline \hline
  Matrix-Product Code & Quantum Parameters  \\
  \hline
   $[C_1,C_1] \cdot A$ & $[[162,158, \geq 2]]_3$     \\
 \hline
  $[C_1,C_2] \cdot A$ & $[[162,146, \geq 4]]_3$     \\
  \hline
  $D$ & $[[162,155, \geq 3]]_3$     \\
   \hline
  \hline
\end{tabular}
\captionof{table}{Stabilizer codes over $\mathbb{F}_3$ by using Theorem \ref{te: QMPCsobreF3}}
\label{tablap17}
\end{center}


 We conclude this section giving parameters of several stabilizer codes over $\mathbb{F}_3$ obtained from subfield-subcodes of affine variety codes. As we have mentioned, we essentially consider Euclidean inner product and our parameters improve some of those given in \cite{lag1}, where the same inner product is used.
 With notations as in Section \ref{subfield}, setting $p=3, r=4, s=1$ and $N_1=80$ and using suitable sets $\Delta$ and Corollary \ref{ham}, we get stabilizer codes with parameters $[[80, 72, \geq 3]]_3$, $[[80,64 , \geq 4]]_3$, $[[80, 56, \geq 6]]_3$ and $[[80,48 , \geq 7]]_3$. In  similar way, with $p=3, r=6, s=1$ and $N_1=728$, we get a $[728,718, \geq 3]]_3$ stabilizer code. Considering Hermitian inner product, the parameters of the codes with length $80$ can be improved \cite[Table I]{lag3}. M. Grassl communicated the authors the existence of a code with parameters $[728,720, \geq 3]]_3$ derived from a Hamming code over $\mathbb{F}_9$.

To finish, we show quantum parameters and defining sets of a couple of codes over $\mathbb{F}_3$ of length $242$, improving \cite{lag1}, for which we do not know stabilizer codes better than them. Consider $p=3,r=5,s=1$ and $N_1=242$ and the corresponding tables for the supporting affine variety stabilizer codes and the codes improving \cite{lag1} are the following:
\vspace{2mm}
\begin{center}
\begin{tabular}{||c|c|c|c|c||}
  \hline \hline
  Code  & $n$ & $k$ & $d \geq$  & Defining set $\Delta$ \\
  \hline \hline
  $C_1$  & 242 & 222 & 4 &   $ \{120,118,112,94,40,75,225,191,89,25
\}$ \\
  \hline
    $C_2$ &242 & 212 & 5 &   $\{ 120,118,112,94,40,75,225,191,89,25, 21,63,189,83,7 \}$ \\
 \hline
$C_3$ & 242 & 202 & 6 &  $\{120,118,112,94,40,75,225,191,89,25, $\\
     & &  &  &
   $ 21,63,189,83,7, 150,208,140,178,50\}$ \\
\hline
 \hline
\end{tabular}
\captionof{table}{Stabilizer affine variety  codes over $\mathbb{F}_3$}
\end{center}

\begin{center}
\begin{tabular}{||c|c|c|c||}
  \hline \hline
  Code  & $n$ & $k$ & $d \geq$  \\
  \hline \hline
   $ \mathrm{SE}(C_2,C_1)$& 242 & 217 & 5   \\
 \hline
  $ \mathrm{SE} (C_3,C_1)$ & 242 & 212 & 6   \\
  \hline
  \hline
\end{tabular}
\captionof{table}{Stabilizer codes over $\mathbb{F}_3$ improving \cite{lag1}}
\end{center}

\section{Quantum stabilizer codes over $\mathbb{F}_q$: $q\neq 2,3$}
\label{no2ni3}

In this section, we are going to show parameters for some new and good stabilizer codes over certain finite fields $\mathbb{F}_q$, with $q\neq 2,3$. To get orthogonal non-singular by columns matrices  of small size, we have developed a MAGMA function to look for matrix-product codes that produce good stabilizer codes. For a start we are going to consider matrices of size $3\times 3$  over the fields $\mathbb{F}_{4}$, $\mathbb{F}_{5}$  and $\mathbb{F}_{7}$.

Set $q=4$, there are 52 orthogonal $3\times 3$  matrices over $\mathbb{F}_{4}$, but only four of them are non-singular by columns. They are
\[
\begin{pmatrix}
1 & a^2 & a^2 \\
a & 0 & a^2  \\
a & a & 1
\end{pmatrix},
\begin{pmatrix}
1 & a & a \\
a^2 & 0 & a  \\
a^2 & a^2 & 1
\end{pmatrix},
\begin{pmatrix}
a & a & 1 \\
a & 0 & a^2  \\
1 & a^2 & a^2
\end{pmatrix},
\begin{pmatrix}
a^2 & a^2 & 1 \\
a^2 & 0 & a  \\
1 & a & a
\end{pmatrix},
\]
$a$ being a primitive element of the field $\mathbb{F}_{4}$.  For $q=5$  (respectively, $q=7$), we can say that one can find 104 (respectively, 304) orthogonal matrices over $\mathbb{F}_{5}$ (respectively, $\mathbb{F}_7$), $64$ (respectively, 96) of them are non-singular by columns. As an example of a matrix over $\mathbb{F}_{5}$ (respectively, $\mathbb{F}_{7}$) of the last type, we have
\[
\begin{pmatrix}
1 & 1 & 2 \\
2 & 1 & 1  \\
1 & 2 & 1
\end{pmatrix}, \;\;\; \left(\mathrm{respectively,} \begin{pmatrix}
2 & 3 & 3 \\
1 & 3 & 1  \\
3 & 3 & 2
\end{pmatrix} \right).
\]

In addition, it is not difficult to check that the only non-singular by columns orthogonal  matrices of size $2\times 2$ over  $\mathbb{F}_{4}$ are:
\[
\begin{pmatrix}
a^2 & a  \\
a & a^2  \\
\end{pmatrix},
\begin{pmatrix}
a & a^2  \\
a^2 & a  \\
\end{pmatrix}.
\]

The following table contains parameters of stabilizer codes obtained with matrix-product codes of Reed-Muller or hyperbolic ones with $m=2$. These codes are over the fields $\mathbb{F}_{4}$, $\mathbb{F}_{5}$ or $\mathbb{F}_{7}$, and we have used non-singular by columns matrices, as above given, and Theorem \ref{UNO} and Corollary \ref{ham}.\vspace{1mm}

\begin{center}
\begin{tabular}{||c|c|c|c||c|c|c|c||}
  \hline \hline
 $n$ & $k$ & $ d \geq $ & Field &  $n$ & $k$ & $d \geq$ & Field \\
  \hline \hline
 48 & 46 & 2  & $\mathbb{F}_4$ &  48 & 42 & 3 & $\mathbb{F}_4$\\
  \hline
 75& 73 & 2  & $\mathbb{F}_5$ &  75 & 70 & 3 & $\mathbb{F}_5$\\
\hline
 75 & 64 & 4  & $\mathbb{F}_5$ &  147 & 145 & 2 & $\mathbb{F}_7$\\
\hline
 147& 142 & 3  & $\mathbb{F}_7$ &    147 & 136 & 4 & $\mathbb{F}_7$\\
  \hline
  \hline
\end{tabular}
\captionof{table}{Quantum codes derived from matrix-product codes}
\label{era7}
\end{center}

Every code in Table \ref{era7} marked with distance larger than or equal to 2 or 3 exceeds the Gilbert-Varshamov bounds \cite{eck,mat,feng}, \cite[Lemma 31]{kkk} in the sense above explained. Our codes $[[75, 70, \geq 3]]_5$ and $[[75, 64, \geq 4]]_5$ have better relative parameters than some showed in \cite{lag2} whose parameters are $[[71, 61, \geq 3]]_5$ and $[[71, 51, \geq 4]]_5$.

We look for good new stabilizer  codes over the mentioned fields obtained from self-orthogonal subfield-subcodes of affine variety codes.  Table \ref{aa} shows  quantum parameters we have obtained from certain affine variety codes $C_i$. The defining subsets $\Delta_i$ are showed in Tables \ref{bb} and \ref{bbb}. Table \ref{cc} contains parameters obtained as described in Theorem \ref{UNO} of stabilizer codes defined by using matrix-product codes with the self-orthogonal constituent codes in Table \ref{aa} and matrices as in the beginning of this section.

\begin{center}
\begin{tabular}{||c|c|c|c|c||c|c|c|c|c||}
  \hline \hline
  Code / Subset & $n$ & $k$ & $d \geq $ & Field & Code / Subset & $n$ & $k$ & $d \geq $ & Field  \\
  \hline \hline
  $C_1$ / $\Delta_{1}$ & 63 & 49 & 4 & $\mathbb{F}_4$ &   $C_2$ / $\Delta_{2}$ & 63 & 43 & 6 & $\mathbb{F}_4$ \\
  \hline
    $C_3$ / $\Delta_{3}$ & 63 & 37 & 7 & $\mathbb{F}_4$ &   $C_4$ / $\Delta_{4}$ & 63 & 31 & 8 & $\mathbb{F}_4$  \\
  \hline
   $C_5$ /  $\Delta_{5}$ & 63 & 25 & 9 & $\mathbb{F}_4$ &  $C_6$ / $\Delta_{6}$ & 496 & 496 & 1 & $\mathbb{F}_5$  \\
  \hline
   $C_7$ /  $\Delta_{7}$ & 496 & 494 & 2  & $\mathbb{F}_5$ &   $C_8$ / $\Delta_{8}$ & 496 & 486 & 3 & $\mathbb{F}_5$ \\
  \hline
    $C_9$ / $\Delta_{9}$ & 496 & 480 & 4 & $\mathbb{F}_5$ &   $C_{10}$ / $\Delta_{10}$ & 96 & 96 & 1& $\mathbb{F}_5$ \\
  \hline
    $C_{11}$ / $\Delta_{11}$ & 96  & 94 & 2  & $\mathbb{F}_5$ &  $C_{12}$ /$\Delta_{12}$ &  96 & 88 & 3 & $\mathbb{F}_5$
      \\
      \hline
    $C_{13}$ / $\Delta_{13}$ & 96  & 84 & 4  & $\mathbb{F}_5$ & $C_{14}$ /$\Delta_{14}$ &  124 & 116 & 3
     & $\mathbb{F}_5$ \\
\hline
$C_{15}$ / $\Delta_{15}$ & 124  & 110 & 4  & $\mathbb{F}_5$ & $C_{16}$ /$\Delta_{16}$ &  124 & 104 & 5 & $\mathbb{F}_5$ \\
     \hline
 $C_{17}$ / $\Delta_{17}$ & 624  & 614 & 3  & $\mathbb{F}_5$ & $C_{18}$ /$\Delta_{18}$ &  624 & 610 & 4 & $\mathbb{F}_5$ \\
    \hline
 $C_{19}$ /$\Delta_{19}$ &  342 & 328 & 4 & $\mathbb{F}_7$ & $C_{20}$ / $\Delta_{20}$ & 342  & 322 & 5  & $\mathbb{F}_7$  \\
 \hline
 $C_{21}$ / $\Delta_{21}$ & 144  & 144 & 1  & $\mathbb{F}_7$ & $C_{22}$ /$\Delta_{18}$ &  144 & 142 & 2 & $\mathbb{F}_7$ \\
 \hline
 $C_{23}$ / $\Delta_{23}$ & 144  & 136 & 3  & $\mathbb{F}_7$ & $C_{24}$ /$\Delta_{24}$ &  144 & 132 & 4 & $\mathbb{F}_7$ \\
  \hline
  \hline
\end{tabular}
\captionof{table}{Quantum codes using affine variety codes}
\label{aa}
\end{center}
\vspace{2mm}

Finally, we use Corollary \ref{ham} for getting better stabilizer codes. The reader can find their parameters in Table \ref{dd}. Comparing with \cite[Table III]{lag3}, we obtain a new code $[[63,31, \geq 9]]_4$ and the parameters of our remaining  codes of length 63 coincide with those in \cite[Table III]{lag3}. Parameters of our codes on $\mathbb{F}_5$ of length 124 also coincide with \cite{lag3} but our code $[[624,612,4]]_5$ is better than $[[624,610,4]]_5$ in \cite{lag3}. Generally speaking we get the same parameters as in \cite[Table III]{lag3} and, occasionally, improve them. Futhermore, we also show good codes with lengths that cannot be reached in \cite{lag3} and they either exceed the Gilbert-Varshamov bounds or improve \cite{edel} or satisfy both conditions.
\vspace{2mm}

\begin{center}
\begin{tabular}{||c|c|c|c|c|c|c||}
  \hline \hline
 Subset & $p$ & $r$ & $s$ & $N_1$  & $N_2$ & $N_3$ \\
  \hline \hline
  $
  \begin{array}{r}
  \Delta_1= \{ 42,44, 50, 11, 46, 58, 43
 \}\\
   \end{array}
   $
    & 2 & 6 & 2 & 63 & - &-\\
  \hline
  $
  \begin{array}{r}
   \Delta_2= \{ 42,44, 50, 11, 46, 58, 43, 41, 38, 26 \}\\
   \end{array}
   $
    & 2 & 6 & 2 & 63 & - &-\\
   \hline
  $
  \begin{array}{r}
   \Delta_3= \{42,44, 50, 11, 46, 58, 43, 41, 38, 26, 57, 39, 30  \}\\
   \end{array}
   $
     & 2 & 6 & 2 & 63 & -&-\\
  \hline
  $
  \begin{array}{r}
   \Delta_4= \{42,44, 50, 11, 46, 58, 43, 41, 38, 26,\\ 57, 39, 30, 60, 51, 15  \}
   \end{array}
   $
     & 2 & 6 & 2 & 63 & - &-\\
    \hline
    $
  \begin{array}{r}
   \Delta_5= \{ 42,44, 50, 11, 46, 58, 43, 41, 38, 26,\\ 57, 39, 30, 60, 51, 15, 45, 54, 27 \}\\
   \end{array}
   $
     & 2 & 6 & 2 & 63 & -&-\\
  \hline
  \hline
   \end{tabular}
   \captionof{table}{Defining sets of affine variety codes}
\label{bb}
\end{center}
\vspace{5mm}

\begin{center}
\begin{tabular}{||c|c|c|c|c|c|c||}
  \hline \hline
 Subset & $p$ & $r$ & $s$ & $N_1$  & $N_2$ & $N_3$  \\
  \hline \hline
    $
  \begin{array}{r}
   \Delta_6=  \emptyset \\
   \end{array}
   $
    & 5 & 3& 1 & 31 & 4 & 4\\\hline
    $
  \begin{array}{r}
   \Delta_7= \{ (0,1,3) \}\\
   \end{array}
   $
    & 5 & 3& 1 & 31 & 4 &4\\\hline
    $
  \begin{array}{r}
   \Delta_8= \{(0, 1, 3 ), (0, 3, 2), (9, 2, 3), (14, 2, 3), (8, 2, 3)\}\\
   \end{array}
   $
    & 5 & 3& 1 & 31 & 4& 4 \\\hline
    $
  \begin{array}{r}
   \Delta_9= \{ (0, 1, 3 ), (0, 3, 2), (9, 2, 3), (14, 2, 3), (8, 2, 3), \\ (23, 2, 0),
   (22, 2, 0 ), (17, 2, 0)\}\\
   \end{array}
   $
    & 5 & 3& 1 & 31 & 4 &4\\\hline
    $
  \begin{array}{r}
   \Delta_{10}= \emptyset \\
   \end{array}
   $
    & 5 & 6& 1 & 24 & 4 &-\\
   \hline
    $
  \begin{array}{r}
   \Delta_{11}= \{ (6,2)\} \\
   \end{array}
   $
    & 5 & 6& 1 & 24 & 4 &-\\
   \hline
    $
  \begin{array}{r}
   \Delta_{12}= \{ (6,2), (12, 1), (5, 2), (1, 2)\} \\
   \end{array}
   $
    & 5 & 6& 1 & 24 & 4 &-\\
    \hline
    $
  \begin{array}{r}
   \Delta_{13}= \{ (6,2), (12, 1), (5, 2), (1, 2), (17, 3), (13, 3)\}\\
   \end{array}
   $
    & 5 & 6& 1 & 24 & 4 &-\\
    \hline
  $
  \begin{array}{r}
   \Delta_{14}= \{ (29, 0 ), ( 21, 0 ), ( 12, 0 ), ( 0, 3) \}\\
   \end{array}
   $
    & 5 & 3 & 1 & 31 & 4&-\\
   \hline
  $
  \begin{array}{r}
   \Delta_{15}= \{ (24, 0 ), (27, 0 ), (11, 0),
(29, 0 ), ( 21, 0 ), ( 12, 0 ), ( 0, 3)\}\\
   \end{array}
   $
     & 5 & 3 & 1 & 31 & 4 &-\\
  \hline
  $
  \begin{array}{r}
   \Delta_{16}= \{ (10, 1 ), (19, 1 ), (2, 1),  (24, 0 ),\\ (27, 0 ), (11, 0),
(29, 0 ), ( 21, 0 ), ( 12, 0 ), ( 0, 3)\}
   \end{array}
   $
     & 5 & 4 & 1 & 624 & - &-\\
    \hline
    $
  \begin{array}{r}
   \Delta_{17}= \{ 156, 295, 227, 511, 59 \}\\
   \end{array}
   $
   & 5 & 4 & 1 & 624 &- &- \\
 \hline
    $
  \begin{array}{r}
   \Delta_{18}= \{ 156, 295, 227, 511, 59, 130, 26 \}\\
   \end{array}
   $
   & 5 & 4 & 1 & 624 &- &- \\
   \hline
    $
  \begin{array}{r}
   \Delta_{19}= \{ 57, 176,206,74, 331,265,145 \}\\
   \end{array}
   $
    & 7 & 3& 1 & 342 & -&- \\\hline
    $
  \begin{array}{r}
   \Delta_{20}= \{57, 176,206,74, 331,265,145, 252,54,36\}\\
   \end{array}
   $
    & 7 & 3& 1 & 342 & -&- \\
   \hline
    $
  \begin{array}{r}
   \Delta_{21}= \emptyset \\
   \end{array}
   $
    & 7 & 2& 1 & 48 & 3&- \\\hline
    $
  \begin{array}{r}
   \Delta_{22}= \{(40,0)\}\\
   \end{array}
   $
    & 7 & 2& 1 & 48 & 3&- \\\hline
    $
  \begin{array}{r}
   \Delta_{23}= \{(40,0), (35, 0), (5, 0), (16, 2)\}\\
   \end{array}
   $
    & 7 & 2& 1 & 48 & 3 &- \\\hline
    $
  \begin{array}{r}
   \Delta_{24}= \{(40,0), (35, 0), (5, 0), (16, 2), (24, 1), (32, 2)\}\\
   \end{array}
   $
    & 7 & 2& 1 & 48 & 3&- \\
   \hline
  \hline
   \end{tabular}
   \captionof{table}{Defining sets of affine variety codes, continued}
\label{bbb}
\end{center}

\begin{center}
\begin{tabular}{||c|c||c|c||}
  \hline \hline
  MP Code & Parameters & MP Code & Parameters\\
  \hline \hline
   $ D_1 : = [C_6,C_7] \cdot A$ & $[[992,990, \geq 2]]_5$  & $D_2 : = [C_7,C_8] \cdot A$ & $[[992,980, \geq 3]]_5$  \\
  \hline
  $D_3 : = [C_7,C_9] \cdot A$ & $[[992,974,\geq 4]]_5$  & $D_4 : = [C_{10},C_{10},C_{10}] \cdot A$ & $[[288,286, \geq 2]]_5$   \\
  \hline
  $D_5 : =[C_{10},C_{11},C_{12}] \cdot A$ & $[[288,278, \geq 3]]_5$  & $D_6 : =[C_{11},C_{11},C_{13}] \cdot A$ & $[[288,272,\geq 4]]_5$  \\
  \hline
  $D_7 : = [C_{21},C_{22}]  \cdot A$ & $[[288,286,\geq 2]]_7$   & $D_8 : = [C_{22},C_{23}] \cdot A$ & $[[288,278,\geq 3]]_7$  \\
  \hline
   $D_9 := [C_{21},C_{21}, C_{22}] \cdot A$ & $[[432,430,\geq 2]]_7$  & $D_{10} : = [C_{21},C_{22},C_{23}] \cdot A$ & $[[432,422, \geq 3]]_7$   \\
  \hline
   $D_{11}:=[C_{22},C_{22},C_{23}] \cdot A$ & $[[432,416, \geq 4]]_7$  &  &  \\
  \hline
  \hline
\end{tabular}
\captionof{table}{Stabilizer codes coming from matrix-product codes}
\label{cc}
\end{center}

\begin{center}
\begin{tabular}{||c|c||c|c||}
  \hline \hline
  Corollary \ref{ham} & Paramameters & Corollary \ref{ham} & Paramarameters\\
  \hline \hline
 $\mathrm{SE }(C_2,C_1) $ & $[[63,46,\geq 5]]_4$  &  $\mathrm{SE }(C_4,C_2)$ & $[[63,37,\geq 8]]_4$  \\
  \hline
  $\mathrm{SE } (C_5,C_3)$ & $[[63,31,\geq 9]]_4$  &  $\mathrm{SE }(C_8,C_7)$ & $[[496,490,\geq 3]]_5$   \\
  \hline
$\mathrm{SE }(C_9,C_8)$ & $[[496,483,\geq 4]]_5$  & $\mathrm{SE } (D_2,D_1)$ & $[[992,985,\geq 3]]_5$  \\
  \hline
  $\mathrm{SE } (D_3,D_2)$ & $[[992,977,\geq 4]]_5$  &    $\mathrm{SE } (C_{12},C_{11})$ & $[[96,86, \geq 4]]_5$ \\
  \hline
  $\mathrm{SE }(D_{5},D_{4})$ & $[[288,283,\geq 3]]_5$   &  $ \mathrm{SE }(D_{6},D_{5})$ & $[[288,275, \geq 4]]_5$ \\
  \hline
  $\mathrm{SE } (C_{15},C_{14})$ & $[[124,113,\geq 4]]_5$ &  $\mathrm{SE }(C_{16},C_{15})$ & $[[124,107,\geq 5]]_5$  \\
  \hline
  $\mathrm{SE } (C_{18},C_{17})$ & $[[624,612, \geq 4]]_5$ &   $\mathrm{SE } (C_{20},C_{19})$ & $[[342,325,\geq 5]]_7$   \\
  \hline
  $\mathrm{SE } (C_{22},C_{23})$ & $[[144,139,\geq 3]]_7$  & $\mathrm{SE } (C_{23},C_{24})$ & $[[144,134, \geq 4]]_7$  \\
  \hline
  $\mathrm{SE } (D_{8},D_{7})$ & $[[288,282, \geq 3]]_7$  & $\mathrm{SE } (D_{10},D_{9})$ & $[[432,426,\geq 3]]_7$  \\
  \hline
  $\mathrm{SE } (D_{11},D_{10})$ & $[[432,419,\geq 4]]_7$  &  &\\
  \hline
  \hline
\end{tabular}
\captionof{table}{Stabilizer codes coming from matrix-product codes and Corollary \ref{ham}}
\label{dd}
\end{center}

\section{Conclusion}

We present new quantum stabilizer codes, our codes are obtained from algebraic linear codes using the CSS code construction and Steane's enlargement. We improve some binary codes of lengths $127$ and $128$ given in \cite{codet} and provide non-binary codes with parameters better than or equal to those in \cite[Table III]{lag3} and others whose lengths cannot be attained with the procedures in \cite{lag3,lag1}. In a future paper, we expect to obtain good codes using the same construction with respect to the Hermitian inner product.

\section*{Acknowledgements}

The authors thank Olav Geil and Markus Grassl for helpful comments.



\end{document}